\documentclass[12pt]{amsart}

\usepackage{amsmath,amssymb,amsthm,textcomp}
\usepackage{amsfonts,graphicx}
\usepackage[mathscr]{eucal}
\usepackage{color}

\theoremstyle{definition}

\numberwithin{equation}{section}

\newcommand{\ncom}{\newcommand}

\ncom{\beq}{\begin{equation}}
\ncom{\eeq}{\end{equation}}
\ncom{\bea}{\begin{eqnarray*}}
\ncom{\eea}{\end{eqnarray*}}
\ncom{\beqa}{\begin{eqnarray}}
\ncom{\eeqa}{\end{eqnarray}}
\ncom{\nno}{\nonumber}
\ncom{\non}{\nonumber}
\ncom{\ds}{\displaystyle}
\ncom{\half}{\frac{1}{2}}
\ncom{\mbx}{\makebox{.25cm}}
\ncom{\hs}{\mbox{\hspace{.25cm}}}
\ncom{\rar}{\rightarrow}
\ncom{\Rar}{\Rightarrow}
\ncom{\noin}{\noindent}
\ncom{\bc}{\begin{center}}
\ncom{\ec}{\end{center}}
\ncom{\sz}{\scriptsize}
\ncom{\rf}{\ref}
\ncom{\s}{\sqrt{2}}
\ncom{\sgm}{\sigma}
\ncom{\Sgm}{\Sigma}
\ncom{\psgm}{\sigma^{\prime}}
\ncom{\dt}{\delta}
\ncom{\Dt}{\Delta}
\ncom{\lmd}{\lambda}
\ncom{\Lmd}{\Lambda}
\ncom{\Th}{\Theta}
\ncom{\e}{\eta}
\ncom{\eps}{\epsilon}
\ncom{\pcc}{\stackrel{P}{>}}
\ncom{\lp}{\stackrel{L_{p}}{>}}
\ncom{\dist}{{\rm\,dist}}
\ncom{\sspan}{{\rm\,span}}
\ncom{\re}{{\rm Re\,}}
\ncom{\im}{{\rm Im\,}}
\ncom{\sgn}{{\rm sgn\,}}
\ncom{\ba}{\begin{array}}
\ncom{\ea}{\end{array}}
\ncom{\hone}{\mbox{\hspace{1em}}}
\ncom{\htwo}{\mbox{\hspace{2em}}}
\ncom{\hthree}{\mbox{\hspace{3em}}}
\ncom{\hfour}{\mbox{\hspace{4em}}}
\ncom{\vone}{\vskip 2ex}
\ncom{\vtwo}{\vskip 4ex}
\ncom{\vonee}{\vskip 1.5ex}
\ncom{\vthree}{\vskip 6ex}
\ncom{\vfour}{\vspace*{8ex}}
\ncom{\norm}{\|\;\;\|}
\ncom{\integ}[4]{\int_{#1}^{#2}\,{#3}\,d{#4}}
\ncom{\vspan}[1]{{{\rm\,span}\{ #1 \}}}
\ncom{\dm}[1]{ {\displaystyle{#1} } }
\ncom{\ri}[1]{{#1} \index{#1}}

\newtheorem{theorem}{\bf Theorem}[section]
\newtheorem{remark}{\bf Remark}[section]
\newtheorem{proposition}{Proposition}[section]
\newtheorem{lemma}{Lemma}[section]

\newtheorem{example}{Example}[section]
\newtheorem{definition}{Definition}[section]
\newtheoremstyle
    {remarkstyle}
    {}
    {11pt}
    {}
    {}
    {\bfseries}
    {:}
    {     }
    {\thmname{#1} \thmnumber{#2} }

\theoremstyle{remarkstyle}

\def\eps{\varepsilon}

\begin{document}
\title{\Large O\lowercase{n the occurrence of boundary solutions in two-way incomplete tables}}
\author[Sayan Ghosh]{S. Ghosh}
\address{Sayan Ghosh, Department of Mathematics,
 Indian Institute of Technology Bombay, Powai, Mumbai 400076, INDIA.}
 \email{sayang@math.iitb.ac.in}
\author{P. Vellaisamy}
\address{P. Vellaisamy, Department of Mathematics,
Indian Institute of Technology Bombay, Powai, Mumbai 400076, INDIA.}
\email{pv@math.iitb.ac.in}
\thanks{The research of S. Ghosh was supported by UGC, Govt. of India grant F.2-2/98 (SA-I)}
\subjclass[2010]{Primary : 62H17}
\keywords{Incomplete tables; Boundary solutions; Log-linear models; NMAR models.}
\begin{abstract}
\noindent The analysis of incomplete contingency tables is an important problem, which is also of practical interest. In this paper, we consider boundary solutions under nonignorable nonresponse models in two-way incomplete tables with data on both variables missing. We establish a result similar to Park {\it et al.}~(2014) on sufficient conditions for the occurrence of boundary solutions. We also provide a new result, which connects the forms of boundary solutions under various parameterizations of the missing data models. This result helps us to obtain the exact form of boundary solutions in the above tables, which improves a claim made in Baker {\it et al.} (1992) and avoids computational burden. A counterexample is provided to show that the sufficient conditions for the occurrence of boundary solutions are not necessary, thereby disproving a conjecture of Kim and Park (2014). Finally, we establish new necessary conditions for the occurrence of boundary solutions under nonignorable nonresponse models in square two-way incomplete tables, and show that they are not sufficient. These conditions are simple and easy to check as they depend only on the observed cell counts. They are useful and important for model selection also. Some real life data sets are analyzed to illustrate the results.
\end{abstract}

\maketitle
\vspace*{-0.7cm}

\section{Introduction}
Contingency tables with fully observed counts and partially classified margins (nonresponses) are called incomplete tables. The following three types of missing data mechanisms have been proposed in the literature (Little and Rubin (2002)): missing completely at random (MCAR), missing at random (MAR) and not missing at random (NMAR). The missing mechanism is said to be (a) MCAR when missingness is independent of both observed and unobserved data, (b) MAR when missingness depends only on observed data, and (c) NMAR if missingness depends on unobserved data. Nonresponses are called ignorable when the missing data mechanism is MAR or MCAR, and the parameters governing the missing data mechanism are distinct from those to be estimated. They are nonignorable when the missing data mechanism is NMAR. 

Log-linear models have generally been used to study missing data mechanisms in incomplete tables (see Park {\it et al.}~(2014) and references therein). However, under nonignorable models, a boundary solution occurs when the cell probabilities of non-respondents are estimated to be zeros for certain levels of the missing variables. That is, the maximum likelihood estimators (MLE's) of the parameters lie on the boundary of the parameter space. Note that the problem of boundary solutions is an important one as it has serious consequences for statistical inference. For example, the observed counts cannot be reproduced by a perfect fit model (a model for which the estimated expected counts are equal to the observed counts) if boundary solutions occur. This implies that the fit is inadequate and the parameter estimates are imprecise. The log likelihood function is flat and, therefore, convergence of the EM algorithm to the boundary MLE's requires a lot of iterations. Also, the eigenvalues of the covariance matrix are inappropriate (either around zero or negative), which implies some parameter estimates have large estimated standard errors and wide confidence intervals. Hence, it is useful to study various forms of boundary solutions and explore conditions for their occurrence in incomplete tables.

Consider two categorical variables with $I$ and $J$ levels. Then an $I\times J\times 2$ table and an $I\times J\times 2\times 2$ table represent two-way incomplete tables with data on one of the variables and data on both the variables missing respectively. The problem of boundary solutions was first considered by Baker and Laird (1988) who proposed a sufficient condition for their occurrence in a $2\times 2\times 2$ incomplete table. Baker {\it et al.}~(1992) studied the problem for an $I\times J\times 2\times 2$ incomplete table, which has non-monotone missing value patterns. For an $I\times J\times 2$ incomplete table with simple monotone missing value patterns, Smith {\it et al.}~(1999) and Clarke (2002) described the problem geometrically, while Clarke and Smith (2005) discussed properties of MLE's in case of boundary solutions. Park {\it et al.}~(2014) proposed sufficient conditions for the occurrence of boundary solutions under various NMAR models in an $I\times I\times 2\times 2$ incomplete table. Recently, Ghosh and Vellaisamy (2016) provided forms of boundary solutions in arbitrary three-way and $n$-dimensional incomplete tables with one or more variables missing, and also established sufficient conditions for their occurrence under various NMAR models. In this paper, we consider the above and other related issues for an $I\times J\times 2\times 2$ table. Note that a lower dimensional incomplete table is not a special case of a higher dimensional one and hence any result for the former cannot be obtained directly from that for the latter. 

The purpose of this paper is to provide a comprehensive treatment of the problem of boundary solutions in two-way incomplete tables with both variables missing. To this effect, we first introduce some notations and consider various identifiable NMAR log-linear models (Models [M1]-[M5]) for an $I\times J\times 2\times 2$ incomplete table. The problem of boundary solutions, along with their forms under the above models, is discussed in Section 3. We formally define boundary solutions for an $I\times J\times 2\times 2$ incomplete table by extending the definition of Baker and Laird (1988), which are unavailable in the literature. A novel result (Proposition 3.1) is provided, which gives the relationship among forms of boundary solutions according to various parameterizations for the missing data models. This helps us to theoretically justify and deduce the exact boundary solutions in those models directly without having to obtain them empirically (see pp.~39-40 of Park {\it et al.}~(2014)) using the EM algorithm. In Section 4, we illustrate this result using some data analysis examples from Baker {\it et al.} (1992), thereby improving a claim made by them on the forms of boundary solutions in $I\times J\times 2\times 2$ tables, which also eliminates computations.

In Section 5, we provide a result (Theorem 5.1) on sufficient conditions for the occurrence of boundary solutions in the above tables, which is similar to Theorem 1 of Park {\it et al.}~(2014) but proved using direct arguments instead of contrapositive ones used in Park {\it et al.}~(2014). While Park {\it et al.}~(2014) consider only Model [M5] in Theorem 1, we consider Models [M1]-[M5] in Theorem 5.1. A counterexample is provided to show that the sufficient conditions for the occurrence of boundary solutions are not necessary, which refutes a conjecture due to Kim and Park (2014).

Finally, we propose new necessary conditions in Theorem 5.2 for the occurrence of boundary solutions under Models [M1]-[M5] in square two-way incomplete tables, and later show that they are not sufficient through a counterexample. Such conditions do not exist in the literature. Note that these conditions help us to identify the non-occurrence of boundary solutions, which is very useful for fitting appropriate models to the incomplete data (model selection). Also, these conditions involve only the observed cell counts and their sums in the tables, and hence can be easily verified. Section 6 provides some concluding remarks.

\section{NMAR log-linear models}
Suppose $Y_{1}$ and $Y_{2}$ are two categorical variables having $I$ and $J$ levels respectively. For $i=1,2,$ let $R_{i}$ denote the missing indicator for $Y_{i}$ so that $R_{i} = 1$ or $2$ if $Y_{i}$ is observed or unobserved. Then we have an $I\times J\times 2\times 2$ incomplete table, corresponding to $Y_{1},~Y_{2},~R_{1}$ and $R_{2}$, with cell counts ${\bf y} = \{y_{ijkl}\}$ where $1\leq i\leq I,~1\leq j\leq J$ and $1\leq k,l\leq 2.$ The vector of observed counts is ${\bf y_{\textrm{obs}}} = (\{y_{ij11}\},\{y_{i+12}\},\{y_{+j21}\},y_{++22}),$ where $\{y_{ij11}\}$ are the fully observed counts and $\{y_{i+12}\},\{y_{+j21}\},y_{++22}$ are the partially classified counts also known as the supplementary margins. All cell counts are assumed to be positive. The fully observed counts are those for which data on both $Y_{1}$ and $Y_{2}$ is available, while data on at most $Y_{1}$ or $Y_{2}$  is available for the supplementary margins. Note that `+' denotes summation over levels of the corresponding variable. For example, $y_{+j21}$ denotes the number of observations corresponding to $Y_{2}=j$ for which data on $Y_{2}$ is observed but data on $Y_{1}$ is missing. Let ${\bf\pi} = \{\pi_{ijkl}\}$ be the vector of cell probabilities, $\mu = \{\mu_{ijkl}\}$ be the vector of expected counts and $N = \sum_{i,j,k,l}y_{ijkl}$ the total number of cell counts. 
For $I = J = 2$, we have the $2\times 2\times 2\times 2$ incomplete table (Table 1). 
\vone\noindent   
{\bf Table 1}. $2\times 2\times 2\times 2$ Incomplete Table.
\begin{center}
$
\begin{array}{|c|c|cc|c|}\hline
& & R_{2} = 1 & & R_{2} = 2 \\ \hline
& & Y_{2} = 1 & Y_{2} = 2 & Y_{2}~\textrm{missing} \\ \hline
R_{1} = 1 & Y_{1} = 1 & y_{1111} & y_{1211} & y_{1+12} \\
& Y_{1} = 2 & y_{2111} & y_{2211} & y_{2+12} \\ \hline
R_{1} = 2 & Y_{1}~\textrm{missing} & y_{+121} & y_{+221} & y_{++22} \\ \hline
\end{array}
$
\end{center}
We consider Poisson sampling for convenience, that is, $Y_{ijkl}\sim P(\mu_{ijkl})$. The likelihood function of ${\bf\mu}$ is 
\begin{eqnarray}\label{lik}
L({\bf\mu};{\bf y}_{\textrm{obs}}) &=& \frac{e^{-\sum_{i,j,k,l}\mu_{ijkl}}\prod_{i,j}\mu_{ij11}^{y_{ij11}}\prod_{i}\mu_{i+12}^{y_{i+12}}\prod_{j}\mu_{+j21}^{y_{+j21}}\mu_{++22}^{y_{++22}}}{\prod_{i,j,k,l}y_{ijkl}!}
\end{eqnarray}
so that the log-likelihood function of ${\bf\mu}$ is
\begin{eqnarray}\label{eq2.2}
l({\bf\mu};{\bf y}_{\textrm{obs}}) &=& \sum_{i,j}y_{ij11}\log \mu_{ij11} + \sum_{i}y_{i+12}\log \mu_{i+12} + \sum_{j}y_{+j21}\log \mu_{+j21} \nonumber \\ 
& & + y_{++22}\log \mu_{++22} - \mu_{++++} + \Delta,
\end{eqnarray}
where $\Delta$ is independent of $\mu_{ijkl}$'s. 
For an $I\times J\times 2\times 2$ incomplete table, Baker {\it et al.}~(1992) proposed the following log-linear model (with no three-way or four-way interactions):
\begin{eqnarray}\label{eq1}
\log \mu_{ijkl} &=& \lambda + \lambda_{Y_{1}}(i) + \lambda_{Y_{2}}(j) + \lambda_{R_{1}}(k) + \lambda_{R_{2}}(l) + \lambda_{Y_{1}Y_{2}}(i,j)   \nonumber \\
& & + \lambda_{Y_{1}R_{1}}(i,k) + \lambda_{Y_{2}R_{1}}(j,k) + \lambda_{Y_{1}R_{2}}(i,l) + \lambda_{Y_{2}R_{2}}(j,l) + \lambda_{R_{1}R_{2}}(k,l),
\end{eqnarray}
where the sum over any argument of a log-linear parameter is zero, for example, $\sum_{i}\lambda_{Y_{1}Y_{2}}(i,j) = \sum_{j}\lambda_{Y_{1}Y_{2}}(i,j) = 0$. To study the various missing mechanisms of $Y_{1}$ and $Y_{2}$, Baker {\it et al.} (1992) introduced the following notations:
\begin{eqnarray*}
a_{ij} &=& \frac{P(R_{1}=2,R_{2}=1|Y_{1}=i,Y_{2}=j)}{P(R_{1}=1,R_{2}=1|Y_{1}=i,Y_{2}=j)} = \frac{\pi_{ij21}}{\pi_{ij11}} = \frac{\mu_{ij21}}{\mu_{ij11}}, \nonumber \\
b_{ij} &=& \frac{P(R_{1}=1,R_{2}=2|Y_{1}=i,Y_{2}=j)}{P(R_{1}=1,R_{2}=1|Y_{1}=i,Y_{2}=j)} = \frac{\pi_{ij12}}{\pi_{ij11}} = \frac{\mu_{ij12}}{\mu_{ij11}},\nonumber \\
m_{ij11} &=& N\pi_{ij11},~g = \frac{P(R_{1}=1,R_{2}=1|Y_{1}=i,Y_{2}=j)P(R_{1}=2,R_{2}=2|Y_{1}=i,Y_{2}=j)}{P(R_{1}=1,R_{2}=2|Y_{1}=i,Y_{2}=j)P(R_{1}=2,R_{2}=1|Y_{1}=i,Y_{2}=j)}.
\end{eqnarray*}
\begin{remark}\label{rem0}
Under (\ref{eq1}), it can be shown that $a_{ij} = \exp[-2\{\lambda_{R_{1}}(1) + \lambda_{Y_{1}R_{1}}(i,1) + \lambda_{Y_{2}R_{1}}(j,1) + \lambda_{R_{1}R_{2}}(1,1)\}]$ and $b_{ij} = \exp[-2\{\lambda_{R_{2}}(1) + \lambda_{Y_{1}R_{2}}(i,1) + \lambda_{Y_{2}R_{2}}(j,1) + \lambda_{R_{1}R_{2}}(1,1)\}]$. Also, we have $g = \frac{\pi_{ij11}\pi_{ij22}}{\pi_{ij12}\pi_{ij21}} = \frac{\mu_{ij11}\mu_{ij22}}{\mu_{ij12}\mu_{ij21}}$. Hence 
\begin{align*}
\log g &= \log\mu_{ij11} + \log\mu_{ij22} - \log\mu_{ij12} - \log\mu_{ij21} \\
\Rightarrow \log g &= \lambda_{R_1R_2}(1,1) + \lambda_{R_1R_2}(2,2) - \lambda_{R_1R_2}(1,2) - \lambda_{R_1R_2}(2,1)\quad(\textrm{from}~(\ref{eq1})) \\
\Rightarrow \log g &= 4\lambda_{R_1R_2}(1,1) ~(\because \lambda_{R_1R_2}(1,2) = \lambda_{R_1R_2}(2,1) = -\lambda_{R_1R_2}(1,1);\lambda_{R_1R_2}(2,2) = -\lambda_{R_1R_2}(1,2))\\
\Rightarrow g &= \exp[4\lambda_{R_1R_2}(1,1)], 
\end{align*}
which is independent of $i$ and $j$.
\end{remark}
Note that $m_{ij11}=\mu_{ij11}$ and $g$ denotes the odds ratio between the missing indicators of $Y_{1}$ and $Y_{2}$. Also, $\mu_{ij21}=m_{ij11}a_{ij}$, $\mu_{ij12}=m_{ij11}b_{ij}$ and $\mu_{ij22}=m_{ij11}a_{ij}b_{ij}g$. Note that $a_{ij}$ is the conditional odds of $Y_{1}$ being missing given $Y_{2}$ is observed, while $b_{ij}$ is the conditional odds of $Y_{2}$ being missing given $Y_{1}$ is observed. Here, $a_{ij}$ and $b_{ij}$ describe the missing mechanisms of $Y_{1}$ and $Y_{2}$,  respectively. Denote $a_{ij}$ $(b_{ij})$ by $\alpha_{i.}$ $(\beta_{i.})$ or $\alpha_{.j}$ $(\beta_{.j})$ or $\alpha_{..}$ $(\beta_{..})$ if it depends only on $i$ or $j$ or none, respectively. Then we have the following definition.
\begin{definition}\label{def0}
The missing mechanism of $Y_{1}$ under (\ref{eq1}) is NMAR if $a_{ij} = \alpha_{i.}$, MAR if $a_{ij} = \alpha_{.j}$ and MCAR if $a_{ij} = \alpha_{..}$. Similarly, the missing mechanism of $Y_{2}$ is NMAR if $b_{ij} = \beta_{.j}$, MAR if $b_{ij} = \beta_{i.}$ and MCAR if $b_{ij} = \beta_{..}$.
\end{definition}
Using Definition \ref{def0} and the above notations, there are nine possible identifiable models (see pp.~647-648 of Baker {\it et al.}~(1992)) based on different missing mechanisms for $Y_{1}$ and $Y_{2}$. The equivalent log-linear models can be obtained as submodels of (\ref{eq1}). As an example, consider the model $(\alpha_{i.},\beta_{i.})$, for which the missing mechanism is NMAR for $Y_{1}$ and MAR for $Y_{2}$. Using the expressions of $a_{ij}$ and $b_{ij}$ above, the corresponding log-linear model is obtained from (\ref{eq1}) by substituting $\lambda_{Y_{2}R_{1}}(j,k)=\lambda_{Y_{2}R_{2}}(j,l)=0$. The following are the five models when the missing mechanism is NMAR for $Y_{1}$ or $Y_{2}$.
\vone\noindent
1. Model M1 (NMAR for $Y_{1}$, MCAR for $Y_{2}$): 
\begin{equation*}
\log \mu_{ijkl} = \lambda + \lambda_{Y_{1}}(i) + \lambda_{Y_{2}}(j) + \lambda_{R_{1}}(k) + \lambda_{R_{2}}(l) + \lambda_{Y_{1}Y_{2}}(i,j) + \lambda_{Y_{1}R_{1}}(i,k) + \lambda_{R_{1}R_{2}}(k,l)
\end{equation*} 
2. Model M2 (NMAR for $Y_{2}$, MCAR for $Y_{1}$): 
\begin{equation*}
\log \mu_{ijkl} = \lambda + \lambda_{Y_{1}}(i) + \lambda_{Y_{2}}(j) + \lambda_{R_{1}}(k) + \lambda_{R_{2}}(l) + \lambda_{Y_{1}Y_{2}}(i,j) + \lambda_{Y_{2}R_{2}}(j,l) + \lambda_{R_{1}R_{2}}(k,l) 
\end{equation*}
3. Model M3 (NMAR for $Y_{1}$, MAR for $Y_{2}$): 
\begin{equation*}
\log \mu_{ijkl} = \lambda + \lambda_{Y_{1}}(i) + \lambda_{Y_{2}}(j) + \lambda_{R_{1}}(k) + \lambda_{R_{2}}(l) + \lambda_{Y_{1}Y_{2}}(i,j) + \lambda_{Y_{1}R_{1}}(i,k) + \lambda_{Y_{1}R_{2}}(i,l) + \lambda_{R_{1}R_{2}}(k,l) 
\end{equation*}
4. Model M4 (NMAR for $Y_{2}$, MAR for $Y_{1}$):
\begin{equation*}
\log \mu_{ijkl} = \lambda + \lambda_{Y_{1}}(i) + \lambda_{Y_{2}}(j) + \lambda_{R_{1}}(k) + \lambda_{R_{2}}(l) + \lambda_{Y_{1}Y_{2}}(i,j) + \lambda_{Y_{2}R_{1}}(j,k) + \lambda_{Y_{2}R_{2}}(j,l) + \lambda_{R_{1}R_{2}}(k,l)
\end{equation*}
5. Model M5 (NMAR for both $Y_{1}$ and $Y_{2}$): 
\begin{equation*}
\log \mu_{ijkl} = \lambda + \lambda_{Y_{1}}(i) + \lambda_{Y_{2}}(j) + \lambda_{R_{1}}(k) + \lambda_{R_{2}}(l) + \lambda_{Y_{1}Y_{2}}(i,j) + \lambda_{Y_{1}R_{1}}(i,k) + \lambda_{Y_{2}R_{2}}(j,l) + \lambda_{R_{1}R_{2}}(k,l)
\end{equation*} 
Note that for Models [M1]-[M5], there is an  association term between a variable and its missing indicator if the missing mechanism is NMAR for that variable (for example, the term $\lambda_{Y_{1}R_{1}}(i,k)$ in Model [M1]), between a variable and the other missing indicator if the missing mechanism is MAR for that variable (for example, the term $\lambda_{Y_{2}R_{1}}(j,k)$ in Model [M4]) and none if the missing mechanism is MCAR for a variable (for example, $\lambda_{Y_{1}R_{1}}(i,k)$ and $\lambda_{Y_{2}R_{1}}(j,k)$ are absent in Model [M2]). 

\section{Boundary solutions in NMAR models}
In this section, we consider boundary solutions under non-ignorable nonresponse (NMAR) models for an $I\times J\times 2\times 2$ incomplete table. We first define boundary solutions under the above models and then present a result relating the forms of boundary solutions in terms of various parameterizations of the models. 

For an incomplete table, boundary solutions in NMAR models occur when the MLE's of nonresponse cell probabilities are all zeros for certain levels of the missing variables. For an $I\times J\times 2$ incomplete table, where data on only $Y_{2}$ is missing, Baker and Laird (1988) defined boundary solutions in the NMAR model for $Y_{2}$ as $\hat{\pi}_{ij2} = 0$ for at least one pair $(i,j)$. For the same model, Clarke and Smith (2005) showed that boundary solutions are given by $\hat{\pi}_{+j2} = 0$ for at least one and at most $(J-1)$ values of $Y_{2}$. Baker and Laird (1988) defined a nonresponse boundary solution under NMAR models in general to be a stationary point that lies on a boundary of the space of parameters modeling the nonignorable nonresponse. Using this, we may extend their definition to an $I\times J\times 2\times 2$ table as follows.
\begin{definition}\label{def1}
Consider an $I\times J\times 2\times 2$ incomplete table, and let $1\leq i\leq I$, $1\leq j\leq J$ and $k,l = 1,2$. Then we have the following. \\
1. A nonresponse boundary solution under the NMAR models for $Y_{1}$ only, that is, Models [M1] and [M3] is an MLE given by $\hat{\pi}_{ij2l} = 0$ for at least one combination $(i,j,l)$. \\
2. A nonresponse boundary solution under the NMAR models for $Y_{2}$ only, that is, Models [M2] and [M4] is an MLE given by $\hat{\pi}_{ijk2} = 0$ for at least one combination $(i,j,k)$. \\
3. A nonresponse boundary solution under the NMAR model for both $Y_{1}$ and $Y_{2}$, that is, Model [M5] is an MLE given by $\hat{\pi}_{ij2l} = 0$ for at least one combination $(i,j,l)$ or $\hat{\pi}_{ijk2} = 0$ for at least one combination $(i,j,k)$.
\end{definition}
Note that in the literature, boundary solutions have usually been defined in terms of cell probabilities because the cell probabilities are in some sense natural to the model for the incomplete table, whereas the loglinear parameters are not. The next proposition explores the relationships among boundary solutions under Models [M1]-[M5] in terms of MLE's of nonresponse cell probabilities, some specific log-linear parameters and $\alpha_{i.}$ or $\beta_{.j}$ for two-way incomplete tables with both variables missing.  
\begin{proposition}\label{prop1}
For an $I\times J\times 2\times 2$ incomplete table, we have the following. \\
1. For Models [M1] and [M3], if boundary solutions occur, then they are given by $\hat{\lambda}_{Y_{1}R_{1}}(i,2) = -\infty\Leftrightarrow\hat{\pi}_{i+2+} = 0\Leftrightarrow\hat{\alpha}_{i.} = 0$ for at least one and at most $(I-1)$ values of $Y_{1}$. \\
2. For Models [M2] and [M4], if boundary solutions occur, then they are given by $\hat{\lambda}_{Y_{2}R_{2}}(j,2) = -\infty\Leftrightarrow\hat{\pi}_{+j+2} = 0\Leftrightarrow\hat{\beta}_{.j} = 0$ for at least one and at most $(J-1)$ values of $Y_{2}$. \\
3. For Model [M5], if boundary solutions occur, then they are given by $\hat{\lambda}_{Y_{1}R_{1}}(i,2) = -\infty$ or $\hat{\lambda}_{Y_{2}R_{2}}(j,2) = -\infty\Leftrightarrow\hat{\pi}_{i+2+} = 0$ or $\hat{\pi}_{+j+2} = 0\Leftrightarrow\hat{\alpha}_{i.} = 0$ for at least one and at most $(I-1)$ values of $Y_{1}$ or $\hat{\beta}_{.j} = 0$ for at least one and at most $(J-1)$ values of $Y_{2}$. 
\end{proposition}
\begin{proof}
See Appendix A1. 
\end{proof}
From the proof of Proposition \ref{prop1} in Appendix A1, note that the one-to-one relation between the cell probabilities and the log-linear parameters cannot be used to derive the connection between the different forms of boundary solutions. This is because it is not obvious which specific log-linear parameters have infinite MLE's just by noting the zero MLE's of the nonresponse cell probabilities when boundary solutions occur.  

\section{Some examples of boundary solutions in NMAR models}
In this section, we reanalyze some examples in Baker {\it et al.} (1992), illustrating the result in Section 3. We use Proposition \ref{prop1} to investigate a claim made by Baker {\it et al.} (1992) regarding forms and occurrence of boundary solutions in an $I\times J\times 2\times 2$ incomplete table. This improvement is useful as it avoids computation and provides the exact boundary solutions under a NMAR model by simply noting the level (s) of the variable (s) for which the MLE's of the parameters are negative or infinite.

First, we present the correct expression of the likelihood ratio statistic for missing data models in such a table. Consider testing the goodness of fit of a null model (here one of the Models [M1]-[M5]) against the alternative model (perfect fit model). Let $\{\hat{\mu}_{ijkl}\}$ and $\{\tilde{\mu}_{ijkl}\}$ denote the MLE's of the expected counts under a null model and a perfect fit model respectively. Also, let $L_{0}$ and $L_{1}$ denote the log-likelihoods for the null and the alternative models, respectively. Then the likelihood ratio statistic is given by
\begin{eqnarray}\label{likr}
G^{2} &=& -2(L_{0} - L_{1}) \nonumber \\
&=& -2\left[\sum_{i,j}y_{ij11}\ln\left(\frac{\hat{\mu}_{ij11}}{\tilde{\mu}_{ij11}}\right) + \sum_{i}y_{i+12}\ln\left(\frac{\hat{\mu}_{i+12}}{\tilde{\mu}_{i+12}}\right) + \sum_{j}y_{+j21}\ln\left(\frac{\hat{\mu}_{+j21}}{\tilde{\mu}_{+j21}}\right) \right.\nonumber \\
&& \left. + y_{++22}\ln\left(\frac{\hat{\mu}_{++22}}{\tilde{\mu}_{++22}}\right) - \hat{\mu}_{++++} + \tilde{\mu}_{++++}\right] \nonumber \\
&=& -2\left[\sum_{i,j}y_{ij11}\ln\left(\frac{\hat{m}_{ij11}}{y_{ij11}}\right) + \sum_{i}y_{i+12}\ln\left(\frac{\sum_{j}\hat{m}_{ij11}\hat{b}_{ij}}{y_{i+12}}\right) + \sum_{j}y_{+j21}\ln\left(\frac{\sum_{i}\hat{m}_{ij11}\hat{a}_{ij}}{y_{+j21}}\right) \right.\nonumber \\
&& \left. + y_{++22}\ln\left(\frac{\sum_{i,j}\hat{m}_{ij11}\hat{a}_{ij}\hat{b}_{ij}\hat{g}}{y_{++22}}\right)- \sum_{i,j}\hat{m}_{ij11}(1 + \hat{a}_{ij} + \hat{b}_{ij} + \hat{a}_{ij}\hat{b}_{ij}\hat{g}) + N\right].
\end{eqnarray}
Note that the last two terms of (\ref{likr}) are missing in the expression of $G^{2}$ in Baker {\it et al.} (1992) (see p.~646). Observe that in general, $\sum_{i,j}\hat{m}_{ij11}(1 + \hat{a}_{ij} + \hat{b}_{ij} + \hat{a}_{ij}\hat{b}_{ij}\hat{g}) \neq N$, unless the hypothetical (null) model is a perfect fit model for example, in which case $G^{2}=0$.

Using Definition \ref{def0} and the notations in Section 2, Models [M1]-[M5] can be represented as follows $-$ Model [M1]: $(\alpha_{i.},\beta_{..})$, Model [M2]: $(\alpha_{..},\beta_{.j})$, Model [M3]: $(\alpha_{i.},\beta_{i.})$, Model [M4]: $(\alpha_{.j},\beta_{.j})$ and Model [M5]: $(\alpha_{i.},\beta_{.j})$. Accordingly, the expression of $G^{2}$ in (\ref{likr}) for each of the above models may be obtained by making suitable substitutions and using the MLE's in Baker {\it et al.} (1992) (see pp.~647-648). For example, the MLE's under the model $(\alpha_{i.},\beta_{..})$ are
\begin{equation*}
\hat{m}_{ij11} = \frac{y_{ij11}y_{i+1+}y_{++11}}{y_{i+11}y_{++1+}},~\sum_{i}\hat{m}_{ij11}\hat{\alpha}_{i.} = y_{+j21},~\hat{\beta}_{..} = \frac{y_{++12}}{y_{++11}},~\hat{g}=\frac{y_{++11}y_{++22}}{y_{++12}y_{++21}}.
\end{equation*}
Hence, from (\ref{likr}), the likelihood ratio statistic is
\begin{equation*}
G^{2} = -2\left[\sum_{i,j}y_{ij11}\ln\left(\frac{y_{i+1+}y_{++11}}{y_{i+11}y_{++1+}}\right) + \sum_{i}y_{i+12}\ln\left(\frac{y_{i+1+}y_{++12}}{y_{i+12}y_{++1+}}\right)\right].
\end{equation*}

Baker {\it et al.} (1992) mentioned that if any solution $\hat{\alpha}_{i.}$ or $\hat{\beta}_{.j}$ to the systems of equations $\sum_{i}N\hat{\pi}_{ij11}\hat{\alpha}_{i.} = y_{+j21}$ and $\sum_{j}N\hat{\pi}_{ij11}\hat{\beta}_{.j} = y_{i+12}$ respectively is negative, then boundary solutions occur, that is, the MLE lies on the boundary of the parameter space. Closed-form boundary MLE's under Models [M1]-[M5] may then be obtained (see p.~649 of Baker {\it et al.} (1992)) by setting certain parameter estimates ($\hat{\alpha}_{i.}$ or $\hat{\beta}_{.j}$) to 0 in the likelihood equations obtained from (\ref{eq2.2}) for the models. They claimed that counterintuitively, the parameter estimate set to 0 need not be the estimate with a negative value as the solution to the above systems of equations. In particular, for a $2\times 2\times 2\times 2$ incomplete table, they suggested examining both boundaries $\hat{\alpha}_{1.}=0$ and $\hat{\alpha}_{2.}=0$; similarly $\hat{\beta}_{.1}=0$ and $\hat{\beta}_{.2}=0$ to determine the minimum value of $G^{2}$, which corresponds to the MLE. We improve this claim and thereby obviate computations by showing that the MLE indeed always occurs on the specific boundary (level (s) of the variable (s)) for which $\hat{\alpha}_{i.}$ or $\hat{\beta}_{.j}$ is negative. In the next three examples, we use Proposition \ref{prop1} to illustrate this point for Models [M1]-[M5].  
\begin{example}\label{ex1}
Consider the data in Table 2 discussed in Baker {\it et al.}~(1992), which cross-classifies mother's self-reported smoking status ($Y_{1}$) ($Y_{1} = 1(2)$ for smoker (non-smoker)) with newborn's weight ($Y_{2}$) ($Y_{2} = 1(2)$ if weight $< 2500$ grams ($\geq 2500$ grams)). The supplementary margins contain data on only smoking status, data on only newborn's weight and missing data on both variables. 
\vone\noindent   
{\bf Table 2}. Birth weight and smoking: observed counts.
\begin{center}
$
\begin{array}{|c|c|cc|c|}\hline
& & R_{2} = 1 & & R_{2} = 2 \\ \hline
& & Y_{2}=1 & Y_{2}=2 & Y_{2}~\textrm{missing} \\ \hline  
R_{1} = 1 & Y_{1}=1 & 4512 & 21009 & 1049 \\
& Y_{1}=2 & 3394 & 24132 & 1135 \\ \hline
R_{1} = 2 & Y_{1}~\textrm{missing} & 142 & 464 & 1224 \\ \hline
\end{array}
$
\end{center}
Baker {\it et al.} (1992) mentioned that $\hat{\alpha}_{2.} < 0$ is obtained on fitting models [M1], [M3] and [M5] to the data in Table 2. Also, the value of $G^{2}$  corresponding to $\hat{\alpha}_{2.}=0$ is larger than that corresponding to $\hat{\alpha}_{1.}=0$ for all the above models, which is incorrect as shown below. When we fit the same models to the data in Table 2 using the `MASS' package in R software, we obtain $\hat{\alpha}_{1.} = 0.0493$ and $\hat{\alpha}_{2.} = -0.0237$ under Models [M1], [M3] and [M5], that is, boundary solutions occur in each of the models. 

Also, $G^{2} = 55.2198~(12.4682)$ under Model [M1], $G^{2} = 55.2168~(12.4638)$ under Model [M3] and $G^{2} = 55.214~(12.464)$ under Model [M5] when $\hat{\alpha}_{1.} = 0~(\hat{\alpha}_{2.} = 0)$. The $G^{2}$ values for $\hat{\alpha}_{2.} = 0$ upon rounding off in each of the models match those given in Table V of Baker {\it et al.} (1992). Hence, $G^{2}$ is minimum for $\hat{\alpha}_{2.} = 0$ in each case, which implies that boundary solutions are given by $\hat{\alpha}_{2.} = 0$ or equivalently $\hat{\pi}_{2+2+}=0$. This result is consistent with points 1 and 3 of Proposition \ref{prop1}. Further, it is the exact form of boundary solutions that we obtain on fitting Models [M1], [M3] and [M5] to the data in Table 2 using the EM algorithm (see the `ecm.cat' function of `cat' package in R software). 
\end{example}
\begin{example}\label{ex2}
Consider the example given in the last paragraph of p.~646 in Baker {\it et al.} (1992). The model [M1] was fitted to the following data: $y_{1111} = 100$, $y_{1211} = 40$, $y_{2111} = 50$, $y_{2211} = 1000$, $y_{1+12} = 0$, $y_{2+12} = 0$, $y_{+121} = 100$, $y_{+221} = 10$ and $y_{++22} = 0$. They mentioned that though $\hat{\alpha}_{1.} < 0$, $G^{2}$ is minimum for $\hat{\alpha}_{2.} = 0$ implying that the MLE is on the boundary $\hat{\alpha}_{2.} = 0$. However, we obtain $\hat{\alpha}_{1.} = 1.0153~(> 0)$ and $\hat{\alpha}_{2.} = -0.0306$ on fitting Model [M1] to the above data. Also, note that $\hat{g} = \frac{y_{++11}y_{++22}}{y_{++12}y_{++21}}$ (see p.~649 of Baker {\it et al.} (1992)) is undefined since $y_{++12} = 0$. Hence, we introduce the following changes: $y_{1+12} = 1$, $y_{2+12} = 1$ and $y_{++22} = 2$ as shown in Table 3.
\vone\noindent 
{\bf Table 3}.
\begin{center}
$
\begin{array}{|c|c|cc|c|}\hline
& & R_{2} = 1 & & R_{2} = 2 \\ \hline
& & Y_{2}=1 & Y_{2}=2 & Y_{2}~\textrm{missing} \\ \hline  
R_{1} = 1 & Y_{1}=1 & 100 & 40 & 1 \\
& Y_{1}=2 & 50 & 1000 & 1 \\ \hline
R_{1} = 2 & Y_{1}~\textrm{missing} & 100 & 10 & 2 \\ \hline
\end{array}
$
\end{center}
On fitting models [M1], [M3] and [M5] to the data in Table 3, we obtain $\hat{\alpha}_{1.} = 1.0098$ under [M1], and $\hat{\alpha}_{1.} = 1.0153$ under [M3] and [M5], along with $\hat{\alpha}_{2.} = -0.0306$ under all the above models, which implies boundary solutions occur in each case. Also, $G^{2} = 426.1604~(17.4704)$ under Model [M1], $G^{2} = 424.3288~(15.669)$ under Model [M3] and $G^{2} = 424.3188~(15.664)$ under Model [M5] when $\hat{\alpha}_{1.} = 0~(\hat{\alpha}_{2.} = 0)$. Hence, $G^{2}$ is minimum for $\hat{\alpha}_{2.} = 0$ in each model, which implies that boundary solutions are given by $\hat{\pi}_{2+2+}=0$. This result is consistent with points 1 and 3 of Proposition \ref{prop1}. Further, it is the exact form of boundary solutions that we obtain on fitting Models [M1], [M3] and [M5] to the data in Table 3 using the EM algorithm.  
\end{example}
\begin{example}\label{ex3}
Consider the data in Table 2 discussed in Example 1. We introduce the following changes corresponding to supplementary margins in Table 2: $464\rightarrow 700$ and $1135\rightarrow 750$. The modified table is shown in Table 4. 
\vone\noindent   
{\bf Table 4}. Birth weight and smoking: observed counts (modified).
\begin{center}
$
\begin{array}{|c|c|cc|c|}\hline
& & R_{2} = 1 & & R_{2} = 2 \\ \hline
& & Y_{2}=1 & Y_{2}=2 & Y_{2}~\textrm{missing} \\ \hline  
R_{1} = 1 & Y_{1}=1 & 4512 & 21009 & 1049 \\
& Y_{1}=2 & 3394 & 24132 & 750 \\ \hline
R_{1} = 2 & Y_{1}~\textrm{missing} & 142 & 700 & 1224 \\ \hline
\end{array}
$
\end{center}
When we fit the models [M2], [M4] and [M5] to the data in Table 4, we obtain $\hat{\beta}_{.1} = 0.2538$ under [M2], and $\hat{\beta}_{.1} = 0.2543$ under [M4] and [M5] along with $\hat{\beta}_{.2} = -0.0047$ under all the above models, that is, boundary solutions occur in each of the models. Also, $G^{2} = 98.5962~(3.3548)$ under Model [M2], $G^{2} = 96.1622~(0.922)$ under Model [M4] and $G^{2} = 96.162~(0.9276)$ under Model [M5] when $\hat{\beta}_{.1} = 0~(\hat{\beta}_{.2} = 0)$. The $G^{2}$ values in brackets above match those obtained using the EM algorithm. Hence, $G^{2}$ is minimum for $\hat{\beta}_{.2} = 0$ in each case, which implies that boundary solutions are given by $\hat{\beta}_{.2} = 0$ or equivalently $\hat{\pi}_{+2+2}=0$. This result is consistent with points 2 and 3 of Proposition \ref{prop1}. Further, it is the exact form of boundary solutions that we obtain on fitting Models [M2], [M4] and [M5] to the data in Table 4 using the EM algorithm. 
\end{example}

\section{Conditions for the occurrence of boundary solutions}
In this section, we discuss sufficient conditions and also propose necessary conditions for the occurrence of boundary solutions in two-way incomplete tables with both variables missing. We show that the sufficient conditions are not necessary, which disproves a conjecture made by Kim and Park (2014). Further, we prove that the proposed necessary conditions are not sufficient. Both sets of conditions are simple to verify since they involve only the observed cell counts in the tables. The sufficient conditions and the necessary conditions are of practical utility in identifying the occurrence and non-occurrence, respectively of boundary solutions in such tables.       
\subsection{Sufficient conditions for the occurrence of boundary solutions}
Following Park {\it et al.}~(2014), define the four odds based on the observed (joint/marginal) cell counts for any pair $(j,j')$ of $Y_{2}$: 
\begin{equation}\label{eq3.1}
\nu_{i}(j,j') = \frac{\hat{\pi}_{ij11}}{\hat{\pi}_{ij'11}},~\nu_{n}(j,j') = \min_{i}\{\nu_{i}(j,j')\},~\nu_{m}(j,j') = \max_{i}\{\nu_{i}(j,j')\},~\nu(j,j') = \frac{y_{+j21}}{y_{+j'21}}.
\end{equation}
Similarly, for a given pair $(i,i')$ of $Y_{1}$, define the four odds using the observed cell counts: 
\begin{equation}\label{eq3.2}
\omega_{j}(i,i') = \frac{\hat{\pi}_{ij11}}{\hat{\pi}_{i'j11}},~\omega_{n}(i,i') = \min_{j}\{\omega_{j}(i,i')\},~\omega_{m}(i,i') = \max_{j}\{\omega_{j}(i,i')\},~\omega(i,i') = \frac{y_{i+12}}{y_{i'+12}}.
\end{equation}
Note that $\nu_{i}(j,j')$ and $\omega_{j}(i,i')$ are called the response odds, while $\nu(j,j')$ and $\omega(i,i')$ are called the nonresponse odds. Using the MLE's of $\{\pi_{ij11}\}$ under Models [M1]-[M5] (see pp.~647-648 of Baker {\it et al.}~(1992)), we deduce that $\nu_{i}(j,j') = \frac{y_{ij11}}{y_{ij'11}}$ and $\omega_{j}(i,i') = \frac{y_{ij11}}{y_{i'j11}}$, which involve only the fully observed counts.   

Theorem 1 of Park {\it et al.}~(2014) deals with sufficient conditions for the occurrence of boundary solutions only under Model [M5]. However, in the next result, we provide such conditions for the occurrence of boundary solutions under Models [M1]-[M5]. Also, we provide a proof which is similar to that of Theorem 1 of Park {\it et al.}~(2014), but we give direct arguments, which are different from the contrapositive ones used by Park {\it et al.}~(2014).
\begin{theorem}\label{th1}
Consider the following conditions for an $I\times I\times 2\times 2$ contingency table. 
\begin{enumerate}
\item[1.] $\nu(j,j')\not\in(\nu_{n}(j,j'),\nu_{m}(j,j'))$ for at least one pair $(j,j')$ of $Y_{2}$,
\item[2.] $\omega(i,i')\not\in(\omega_{n}(i,i'),\omega_{m}(i,i'))$ for at least one pair $(i,i')$ of $Y_{1}$.
\end{enumerate}
Then we have the following: 
\begin{enumerate}
\item[(a)] Boundary solutions in NMAR models for only $Y_{1}$ (Models [M1] and [M3]) occur if Condition 1 holds. 
\item[(b)] Boundary solutions in NMAR models for only $Y_{2}$ (Models [M2] and [M4]) occur if Condition 2 holds.
\item[(c)] Boundary solutions in the NMAR model for both $Y_{1}$ and $Y_{2}$ (Model [M5]) occur if Condition 1 or Condition 2 holds. 
\end{enumerate}
\end{theorem}
\begin{proof}
See Appendix A2.
\end{proof}  

\subsection{The sufficient conditions are not necessary}
The next example shows that the sufficient conditions for the occurrence of boundary solutions mentioned in Theorem \ref{th1} are not necessary. This result has not been discussed in the literature earlier. In fact, Kim and Park (2014) proved that the above conditions are both necessary and sufficient for a $2\times 2\times 2\times 2$ incomplete table. They conjectured that a similar result would hold for general two-way incomplete tables as well. 

\begin{example}\label{ex5}
Consider Table 5 discussed in Park {\it et al.}~(2014), which cross-classifies data on bone mineral density ($Y_{1}$) and family income ($Y_{2}$) in a $3\times 3\times 2\times 2$ incomplete table. Both variables $Y_{1}$ and $Y_{2}$ have three levels. The total count is $2998$ out of which data on $Y_{1}$ and $Y_{2}$ are available for $1844$ persons, data on $Y_{1}$ only for $231$ persons, data on $Y_{2}$ only for $878$ persons, and data on neither of them for $45$ persons.
\vone\noindent
{\bf Table 5}. Bone mineral density ($Y_{1}$) and family income ($Y_{2}$).
\begin{center}
	$
	\begin{array}{|c|c|ccc|c|}\hline
	& & R_{2} = 1 & & & R_{2} = 2 \\
	& & Y_{2} = 1 & Y_{2} = 2 & Y_{2} = 3 & \text{Missing} \\ \hline
	& Y_{1} = 1 & 621 & 290 & 284 & 135 \\   
	R_{1} = 1 & Y_{1} = 2 & 260 & 131 & 117 & 69 \\
	& Y_{1} = 3 & 93 & 30 & 18 & 27 \\ \hline
	R_{1} = 2 & \text{Missing} & 456 & 156 & 266 & 45 \\ \hline
	\end{array}
	$
\end{center}	
Now, we introduce the following changes corresponding to supplementary margins in Table 5: $266\rightarrow 125,~69\rightarrow 60$ and $27\rightarrow 20$. The modified table is shown in Table 6. 
\vone\noindent
{\bf Table 6}. Modified Table 5.
\begin{center}
$
\begin{array}{|c|c|ccc|c|}\hline
& & R_{2} = 1 & & & R_{2} = 2 \\
& & Y_{2} = 1 & Y_{2} = 2 & Y_{2} = 3 & \text{Missing} \\ \hline
& Y_{1} = 1 & 621 & 290 & 284 & 135 \\   
R_{1} = 1 & Y_{1} = 2 & 260 & 131 & 117 & 60 \\
& Y_{1} = 3 & 93 & 30 & 18 & 20 \\ \hline
R_{1} = 2 & \text{Missing} & 456 & 156 & 125 & 45 \\ \hline
\end{array}
$
\end{center}
From Table 6, $\nu(1,2) = 456/156 = 2.92$, $\nu(1,3) = 456/125 = 3.65$, $\nu(2,3) = 156/125 = 1.25$, $\omega(1,2)= 135/60 = 2.25$, $\omega(1,3) = 135/20 = 6.75$ and $\omega(2,3) = 60/20 = 3.00$. Let $I_{\nu}(j,j') = (\nu_{n}(j,j'),\nu_{m}(j,j'))$ and $I_{\omega}(i,i') = (\omega_{n}(i,i'),\omega_{m}(i,i'))$. Then from Table 6, it can be shown that $\nu(1,2)\in I_{\nu}(1,2)=(260/131,93/30)$, $\nu(1,3)\in I_{\nu}(1,3)=(621/284,93/18)$, $\nu(2,3)\in I_{\nu}(2,3)=(290/284,30/18)$, $\omega(1,2)\in I_{\omega}(1,2)=(290/131,284/117)$, $\omega(1,3)\in I_{\omega}(1,3)=(621/93,284/18)$ and $\omega(2,3)\in I_{\omega}(2,3)=(260/93,117/18)$ so that the sufficient conditions for the occurrence of boundary solutions in Theorem \ref{th1} are not satisfied. The MLE's of the parameters obtained on fitting Models [M1]-[M5] in various subtables of Table 6 are shown in Table 7.
\vone\noindent
{\bf Table 7}. MLE's of parameters in subtables of Table 6.
\begin{center}
$
\begin{array}{|c|c|c|c|c|}\hline
\text{Subtable} & \text{NMAR} & \text{MLE's} & \text{Boundary} \\ 
& \text{model} & & \text{solutions} \\ \hline
Y_{1} & [M1] & \hat{\alpha}_{1.} = 0.6556,\hat{\alpha}_{2.} = -1.0537,\hat{\alpha}_{3.} = 3.4109 & \hat{\pi}_{2+2+} = 0 \\ \hline
Y_{2} & [M2] & \hat{\beta}_{.1} = 0.1355,\hat{\beta}_{.2} = 0.3420,\hat{\beta}_{.3} = -0.1846 & \hat{\pi}_{+3+2} = 0 \\ \hline
Y_{1}Y_{2} & [M1] & \hat{\alpha}_{1.} = 0.6556,\hat{\alpha}_{2.} = -1.0537,\hat{\alpha}_{3.} = 3.4109 & \hat{\pi}_{2+2+} = 0 \\
& [M3] & \hat{\alpha}_{1.} = 0.6534,\hat{\alpha}_{2.} = -1.0551,\hat{\alpha}_{3.} = 3.4874 & \hat{\pi}_{2+2+} = 0 \\ \hline
Y_{1}Y_{2} & [M2] & \hat{\beta}_{.1} = 0.1355,\hat{\beta}_{.2} = 0.3420,\hat{\beta}_{.3} = -0.1846 & \hat{\pi}_{+3+2} = 0 \\ 
& [M4] & \hat{\beta}_{.1} = 0.1421,\hat{\beta}_{.2} = 0.3289,\hat{\beta}_{.3} = -0.1712 & \hat{\pi}_{+3+2} = 0 \\ \hline  
Y_{1}Y_{2} & [M5] & \hat{\alpha}_{1.} =  0.6534,\hat{\alpha}_{2.} = -1.0551,\hat{\alpha}_{3.} = 3.4874, & \hat{\pi}_{2+2+} = 0, \\ 
& & \hat{\beta}_{.1} = 0.1421,\hat{\beta}_{.2} = 0.3289,\hat{\beta}_{.3} = -0.1712 & \hat{\pi}_{+3+2} = 0 \\ \hline    
\end{array}
$
\end{center}
From the above table, note that in each subtable, at least one of $\hat{\alpha}_{i.}$ and $\hat{\beta}_{.j}$ is negative, which imply that boundary solutions occur. The forms of boundary solutions under the Models [M1]-[M5] are also the same as described in Section 3. This shows that for an $I\times J\times 2\times 2$ incomplete table, where $I,J\geq 3$, the sufficient conditions for the occurrence of boundary solutions under Models [M1]-[M5] in Theorem \ref{th1} are not necessary.
\end{example}

\subsection{Necessary conditions for the occurrence of boundary solutions}
We next state below a result due to Kaykobad (1985), which will be used later to obtain a result on the occurrence of boundary solutions. 
\begin{lemma}\label{th3}
Suppose $A = (a_{ij})$ is a matrix with $a_{ij}\geq 0$ for $i\neq j = 1,2,\ldots,n$ and $a_{ii} > 0$. Also, let ${\bf b} = (b_{j})$, where $b_{j} > 0$ for $1\leq j\leq n$. If 
\begin{equation}\label{eq}
b_{i} > \sum_{j\neq i=1}^{n}a_{ij}\frac{b_{j}}{a_{jj}},\quad\forall~1\leq i\leq n, 
\end{equation}
then $A$ is invertible and $A^{-1}{\bf b} > {\bf 0}$.
\end{lemma}      
Using Lemma \ref{th3}, the next result provides necessary conditions for the occurrence of boundary solutions under Models [M1]-[M5] in square two-way incomplete tables.
\begin{theorem}\label{th4}
For an $I\times I\times 2\times 2$ incomplete table, consider the following conditions: 
\begin{enumerate}
\item[1.] $y_{+j21} \leq \sum_{i\neq j=1}^{I}\hat{\mu}_{ji11}\frac{y_{+i21}}{\hat{\mu}_{ii11}}\quad\textrm{for at least one}~j = 1,2,\ldots,I,$
\item[2.] $y_{i+12} \leq \sum_{j\neq i=1}^{I}\hat{\mu}_{ij11}\frac{y_{j+12}}{\hat{\mu}_{jj11}}\quad\textrm{for at least one}~i = 1,2,\ldots,I$,
\end{enumerate}
where $\hat{\mu}_{ij11}$ is the MLE of $\mu_{ij11}$. Also, let $\{\hat{\mu}_{ij11}\} > 0$, $\{y_{i+12}\} > 0$ and $\{y_{+j21}\} > 0$. Then we have the following: 
\begin{enumerate}
\item[(a)] If boundary solutions under Models [M1] and [M3] occur, then only Condition 1 holds. 
\item[(b)] If boundary solutions under Models [M2] and [M4] occur, then only Condition 2 holds. 
\item[(c)] If boundary solutions under the Model [M5] occur, then Condition 1 or Condition 2 holds.
\end{enumerate} 
\end{theorem} 
\begin{proof}
See Appendix A3.
\end{proof}
Henceforth, we denote $A = (a_{ij}) = (\hat{\mu}_{ij11})$, ${\bf b} = (b_{j}) = (y_{+j21})$ and ${\bf b^{\ast}} = (b^{\ast}_{i}) = (y_{i+12})$ for $1\leq i\leq I,~1\leq j\leq I$. The example below is an application of Theorem \ref{th4}. 
\begin{example}\label{ex6}
From Table 6 in Example \ref{ex5}, we have the following:
\begin{equation*}
A = 
\begin{pmatrix}
621 & 290 & 284 \\
260 & 131 & 117 \\
93 & 30 & 18 
\end{pmatrix},\quad {\bf b} = (456, 156, 125),\quad {\bf b^{\ast}} = (135,60,20).
\end{equation*} 
The MLE's ${\bf \hat{\alpha}} = (\hat{\alpha}_{i.})$ and ${\bf \hat{\beta}} = (\hat{\beta}_{.j})$ under Model [M5] satisfy respectively the systems $A^{T}{\bf \alpha} = {\bf b}$ from (\ref{eq4.3}) and $A{\bf \beta} = {\bf b^{\ast}}$ from (\ref{eq4.4}) for $i,j = 1,2,3$. From Table 7, we observe that if Model [M5] is fitted to the data in Table 9, then we obtain $\hat{\alpha}_{2.} < 0$ and $\hat{\beta}_{.3} < 0$, that is, boundary solutions occur. Now we need to verify if both Conditions 1 and 2 of Theorem \ref{th4} hold. For the matrix $A^{T}$ and the vector ${\bf b}$, we have
\begin{eqnarray*}
456 &<& a_{12}\times\frac{b_{2}}{a_{22}} + a_{13}\times\frac{b_{3}}{a_{33}} = 260\times\frac{156}{131} + 93\times\frac{125}{18} = 955.4516, \\
156 &<& a_{21}\times\frac{b_{1}}{a_{11}} + a_{23}\times\frac{b_{3}}{a_{33}} = 290\times\frac{456}{621} + 30\times\frac{125}{18} = 421.2802, \\
125 &<& a_{31}\times\frac{b_{1}}{a_{11}} + a_{32}\times\frac{b_{2}}{a_{22}} = 284\times\frac{456}{621} + 117\times\frac{156}{131} = 347.8693,
\end{eqnarray*}
so that Condition 1 in Theorem \ref{th4} is satisfied. Also, for the matrix $A$ and the vector ${\bf b^{\ast}}$, we have 
\begin{eqnarray*}
135 &<& a_{12}\times\frac{b^{\ast}_{2}}{a_{22}} + a_{13}\times\frac{b^{\ast}_{3}}{a_{33}} = 290\times\frac{60}{131} + 284\times\frac{20}{18} = 448.38, \\
60 &<& a_{21}\times\frac{b^{\ast}_{1}}{a_{11}} + a_{23}\times\frac{b^{\ast}_{3}}{a_{33}} = 260\times\frac{135}{621} + 117\times\frac{20}{18} = 186.5217, \\
20 &<& a_{31}\times\frac{b^{\ast}_{1}}{a_{11}} + a_{32}\times\frac{b^{\ast}_{2}}{a_{22}} = 93\times\frac{135}{621} + 30\times\frac{60}{131} = 33.9578,
\end{eqnarray*}
so that Condition 2 in Theorem \ref{th4} is satisfied. Further, from Table 7, we observe that boundary solutions also occur if Models [M1]-[M4] are fitted to data in Table 6. Then only Condition 1 is satisfied if boundary solutions under [M1] and [M3] occur, while only Condition 2 is satisfied if boundary solutions under [M2] and [M4] occur. This is because the MLE $\hat{\alpha} = (\hat{\alpha}_{i.})$ under Models [M1] and [M3] satisfies the system $A^{T}\alpha = {\bf b}$, while the MLE $\hat{\beta} = (\hat{\beta}_{.j})$ under Models [M2] and [M4] satisfies the system $A\beta = {\bf b^{\ast}}$. 
\end{example}

\subsection{The necessary conditions are not sufficient} 
The next example shows that the necessary conditions for the occurrence of boundary solutions in Theorem \ref{th4} are not sufficient.
\begin{example}\label{ex7}
In Example \ref{ex6}, replace $456$ by $366$ in ${\bf b}$ and $20$ by $15$ in ${\bf b^{\ast}}$ so that ${\bf b} = (366,156,125)$ and ${\bf b^{\ast}} = (135,60,15)$ now. For the matrix $A^{T}$ and the vector ${\bf b}$, we have
\begin{eqnarray*}
366 &<& a_{12}\times\frac{b_{2}}{a_{22}} + a_{13}\times\frac{b_{3}}{a_{33}} = 260\times\frac{156}{131} + 93\times\frac{125}{18} = 955.4516, \\
156 &<& a_{21}\times\frac{b_{1}}{a_{11}} + a_{23}\times\frac{b_{3}}{a_{33}} = 290\times\frac{366}{621} + 30\times\frac{125}{18} = 379.2512, \\
125 &<& a_{31}\times\frac{b_{1}}{a_{11}} + a_{32}\times\frac{b_{2}}{a_{22}} = 284\times\frac{366}{621} + 117\times\frac{156}{131} = 306.7099,
\end{eqnarray*}
so that Condition 1 in Theorem \ref{th4} is satisfied. Also, for the matrix $A$ and the vector ${\bf b^{\ast}}$, we have 
\begin{eqnarray*}
135 &<& a_{12}\times\frac{b^{\ast}_{2}}{a_{22}} + a_{13}\times\frac{b^{\ast}_{3}}{a_{33}} = 290\times\frac{60}{131} + 284\times\frac{15}{18} = 369.4911, \\
60 &<& a_{21}\times\frac{b^{\ast}_{1}}{a_{11}} + a_{23}\times\frac{b^{\ast}_{3}}{a_{33}} = 260\times\frac{135}{621} + 117\times\frac{15}{18} = 154.0217, \\
15 &<& a_{31}\times\frac{b^{\ast}_{1}}{a_{11}} + a_{32}\times\frac{b^{\ast}_{2}}{a_{22}} = 93\times\frac{135}{621} + 30\times\frac{60}{131} = 33.9578,
\end{eqnarray*}
so that Condition 2 in Theorem \ref{th4} is satisfied. Now, when we solve the system $A^{T}{\bf \alpha} = {\bf b}$, then we obtain the MLE's $\hat{\alpha}_{1.} = 0.0133$, $\hat{\alpha}_{2.} = 0.7796$ and $\hat{\alpha}_{3.} = 1.6671$. So, there are no boundary solutions under Model [M3]. Similarly, the system $A{\bf \beta} = {\bf b^{\ast}}$ yields the MLE's $\hat{\beta}_{.1} = 0.041$, $\hat{\beta}_{.2} = 0.3655$ and $\hat{\beta}_{.3} = 0.0126$, that is, there are no boundary solutions under Model [M4]. Since the MLE's in Model [M5] satisfy both the systems $A^{T}{\bf \alpha} = {\bf b}$ and $A{\bf \beta} = {\bf b^{\ast}}$, there are no boundary solutions under [M5] as well. Similar results hold for Models [M1] and [M2]. Hence, the conditions in Theorem \ref{th4} are not sufficient for the occurrence of boundary solutions under Models [M1]-[M5]. 
\end{example}

\subsection{Importance of the necessary conditions}
Here, we discuss additional details about Theorem \ref{th4} and discuss its simplicity and effectiveness. 

From Theorem \ref{th4}, note that if $\{y_{i+12}\}$, $\{y_{+j21}\}$, and/or $\{\hat{\mu}_{ii11}\}$ are large, then Conditions 1 and 2 may not hold. Indeed, if the inequalities in Conditions 1 and 2 are reversed for all $1\leq i\leq I$ and $1\leq j\leq I$, then from statements (a), (b) and (c) of Theorem \ref{th4}, boundary solutions do not occur on fitting Models [M1]-[M5] in an $I\times I\times 2\times 2$ incomplete table. 

It is known that when boundary solutions occur, perfect fit models (here Models M3], [M4] and [M5]) cannot reproduce the observed counts, indicating poor fit and imprecision of the parameter estimates. The MLE's of the parameters under NMAR models lie on the boundary of the parameter space and the log likelihood function tends to be flat, which makes derivation of the MLE's computationally intensive. Also, the corresponding covariance matrix has unreasonable eigenvalues (close to either zero or negative), which implies the estimated standard errors for some parameter estimates are large. Hence, for model selection, we prefer NMAR models which don't yield boundary solutions upon fitting them to the given data. 

Theorem \ref{th1} provides conditions, which help us identify the occurrence of boundary solutions. However, boundary solutions may occur under some NMAR models if any of the sufficient conditions in Theorem \ref{th1} does not hold. This implies that Theorem \ref{th1} cannot always provide us the set of plausible NMAR models for model selection. However, note that Theorem \ref{th4} is very useful in this regard since it gives us an insight into verifying the non-occurrence of boundary solutions under each of the NMAR models [M1]-[M5]. That is, if any of the necessary conditions in Theorem \ref{th4} does not hold, then we know for sure that boundary solutions do not occur. This always helps us to obtain the list of candidate NMAR models suitable for fitting the given data. Hence, Theorem \ref{th4} is more reliable than Theorem \ref{th1} for the purpose of model selection in square two-way incomplete tables.

The non-boundary MLE's of $\mu_{ij11}$ are $\hat{\mu}_{ij11} = \frac{y_{ij11}y_{i+1+}y_{++11}}{y_{i+11}y_{++1+}}$ under Model [M1], $\hat{\mu}_{ij11} = \frac{y_{ij11}y_{+j+1}y_{++11}}{y_{+j11}y_{+++1}}$ under Model [M2], and $\hat{\mu}_{ij11} = y_{ij11}$ under Models [M3], [M4] and [M5] (see pp.~647-648 of Baker {\it et al.} (1992)), which involve only the observed cell counts and their sums. Hence, from Theorem \ref{th4}, there is no need to solve any system of likelihood equations, use the EM algorithm or compute odds (based on the observed (joint/marginal) cell counts) to check for the non-occurrence of boundary solutions in an $I\times I\times 2\times 2$ incomplete table.

\begin{remark}
If $A_{D} = \text{diag}(a_{11},\ldots,a_{II})$, then from Kaykobad (1985), the solutions $\alpha = (\alpha_{i.})$ of the system $A^{T}{\alpha} = {\bf b}$ may be obtained iteratively as follows. 
\begin{eqnarray}\label{eq4.5}
\alpha^{(0)} &=& A_{D}^{-1}{\bf b} \nonumber \\
\alpha^{(n+1)} &=& \alpha^{(n)} + A_{D}^{-1}({\bf b} - A^{T}\alpha^{(n)}),\quad n = 0,1,2,\ldots .
\end{eqnarray}
Similarly, the solutions $\beta = (\beta_{.j})$ of the system $A{\beta} = {\bf b^{\ast}}$ may be obtained iteratively as follows. 
\begin{eqnarray}\label{eq4.6}
\beta^{(0)} &=& A_{D}^{-1}{\bf b^{\ast}} \nonumber \\
\beta^{(n+1)} &=& \beta^{(n)} + A_{D}^{-1}({\bf b^{\ast}} - A\beta^{(n)}),\quad n = 0,1,2,\ldots .
\end{eqnarray}
Both the sequences (\ref{eq4.5}) and (\ref{eq4.6}) converge to the solutions of the respective systems. 
\end{remark}

\section{Conclusions}
In this paper, we have discussed the problem of boundary solutions that occur under various NMAR models for an $I\times J\times 2\times 2$ table. We formally define boundary solutions for such a table and provide a result (Proposition 3.1) that theoretically connects and justifies various forms of these solutions under alternative parametrizations of the missing data models. This eliminates the need of using the EM algorithm (see pp.~39-40 of Park {\it et al.} (2014)) to empirically obtain the forms of the solutions in two-way incomplete tables. The above result is then used to improve a claim in Baker {\it et al.}~(1992) regarding the occurrence of boundary solutions. We give the precise forms of such solutions by just noting the corresponding level (s) of the variable (s) in the table, which reduces computational burden.

As discussed earlier, boundary solutions pose a lot of problems for estimation and inference under NMAR models in incomplete tables. Hence, it is important to investigate sufficient and necessary conditions for their occurrence in such tables. We have provided a result (Theorem 5.1) on sufficient conditions for the occurrence of boundary solutions in an $I\times J\times 2\times 2$ table. While Park et al. (2014) consider only Model [M5], we consider Models [M1]-[M5] in Theorem 5.1. We use a similar approach but give direct arguments instead of contrapositive ones used in Theorem 1 of Park {\it et al.}~(2014) for proving Theorem 5.1. Kim and Park~(2014) conjectured that these conditions would also be necessary for general two-way incomplete tables. However, we show by a counterexample that this is not the case for $I,J\geq 3$, thereby disproving the conjecture. 

We have also established necessary conditions in Theorem 5.2 for the occurrence of boundary solutions in an $I\times J\times 2\times 2$ table, which have not been discussed in the literature so far. As discussed in Section 5.5, these conditions are of practical utility to identify the non-occurrence of boundary solutions and hence for model selection. However, we show by a counterexample that these conditions are not sufficient. Note that a major advantage of the proposed sufficient conditions and necessary conditions is that they depend only on the observed cell counts in the table or their sums. As mentioned in Park {\it et al.}~(2014), this makes the verification process much easier, and avoids using the EM algorithm or solving likelihood equations. Finally, all the above results are illustrated using six data analysis examples. It would be helpful to obtain a set of conditions involving only the observed cell counts, which are sufficient as well as necessary for the occurrence of boundary solutions in two-way incomplete tables with both variables missing.  
\vone\noindent
{\bf Acknowledgements} : The authors are grateful to the referees for carefully reading the manuscript and suggesting numerous improvements.

\section*{Appendix}
\subsection*{A1}{\it Proof of Proposition 3.1:} From Definition \ref{def1}, it follows that if boundary solutions occur under the Models [M1]-[M5], then the MLE's of the cell probabilities except some of the nonresponse ones are all non-zero. On substituting $k=l=1$ (for response cell probabilities) in the above models and using the parameter constraints, we can then deduce that the MLE's of the constant, the main effects and the association terms between $Y_{i}$'s, between $R_{i}$'s, and between $Y_{i}$ and $R_{j}$ for $i\neq j$ are all finite. This is because non-zero terms (response cell probabilities) on the LHS of the log-linear models imply that the log-linear parameters on the RHS are finite. 
\vone\noindent  
Consider part 1 first. For the Models [M1] and [M3], the log-linear parameters modelling the non-ignorable nonresponse (NMAR) mechanism of $Y_{1}$ are $\lambda_{R_{1}}(k)$ and $\lambda_{Y_{1}R_{1}}(i,k)$. If boundary solutions occur, then they are of the form $\hat{\pi}_{ij2l} = 0$ (see point 1 of Definition \ref{def1}), which implies $\hat{\lambda}_{Y_{1}R_{1}}(i,2) = -\infty$ for at least one $i$ since the other parameters are finite as mentioned above. Then under Model [M1], we have
\begin{eqnarray*}
\hat{\pi}_{i+2+} &=& \sum_{j,l}\hat{\pi}_{ij2l} \nonumber \\
&=& \frac{1}{N}\sum_{j,l}\exp\{\hat{\lambda} + \hat{\lambda}_{Y_{1}}(i) + \hat{\lambda}_{Y_{2}}(j) + \hat{\lambda}_{R_{1}}(2) + \hat{\lambda}_{R_{2}}(l) + \hat{\lambda}_{Y_{1}R_{1}}(i,2) + \hat{\lambda}_{Y_{1}Y_{2}}(i,j) \nonumber \\
& & + \hat{\lambda}_{R_{1}R_{2}}(2,l)\} \nonumber \\
&=& 0 \nonumber
\end{eqnarray*}
for at least one $i$. Conversely, we have
\begin{align*}
&\hat{\pi}_{i+2+} = 0 \quad\text{(for at least one $i$)}\nonumber \\
&\Rightarrow \sum_{j,l}\exp\{\hat{\lambda} + \hat{\lambda}_{Y_{1}}(i) + \hat{\lambda}_{Y_{2}}(j) + \hat{\lambda}_{R_{1}}(2) + \hat{\lambda}_{R_{2}}(l) + \hat{\lambda}_{Y_{1}R_{1}}(i,2) + \hat{\lambda}_{Y_{1}Y_{2}}(i,j) + \hat{\lambda}_{R_{1}R_{2}}(2,l)\} = 0 \nonumber \\
&\Rightarrow\hat{\lambda}_{Y_{1}R_{1}}(i,2) = -\infty\quad\text{for at least one $i$}, \nonumber
\end{align*}
so that $\hat{\lambda}_{Y_{1}R_{1}}(i,2) = -\infty\Leftrightarrow\hat{\pi}_{i+2+} = 0$ for at least one $i$ under Model [M1]. The same can be shown for Model [M3]. Under Models [M1] and [M3], $a_{ij} = \exp[2\{\lambda_{R_{1}}(2) + \lambda_{Y_{1}R_{1}}(i,2) + \lambda_{R_{1}R_{2}}(2,1)\}]$. Since $a_{ij}$ depends only on $i$, we have $a_{ij} = \alpha_{i.}$. It is clear that $\hat{\alpha}_{i.} = 0\Leftrightarrow\hat{\lambda}_{Y_{1}R_{1}}(i,2) = -\infty$. Also, note that by definition of $a_{ij}$, if $\hat{\alpha}_{i.} = 0$ for all $1\leq i\leq I$, then $y_{+j21} = 0$ for all $1\leq j\leq J$, which is a contradiction since supplementary margins are assumed to be positive. Hence, under Models [M1] and [M3], boundary solutions are given by $\hat{\lambda}_{Y_{1}R_{1}}(i,2) = -\infty\Leftrightarrow\hat{\pi}_{i+2+} = 0\Leftrightarrow\hat{\alpha}_{i.} = 0$ for at least one and at most $(I-1)$ values of $Y_{1}$. 
\vone\noindent
Consider part 2 now. Under Models [M2] and [M4], the log-linear parameters modelling the NMAR nechanism of $Y_{2}$ are $\lambda_{R_{2}}(l)$ and $\lambda_{Y_{2}R_{2}}(j,l)$. Also, $b_{ij} = \exp[2\{\lambda_{R_{2}}(2) + \lambda_{Y_{2}R_{2}}(j,2) + \lambda_{R_{1}R_{2}}(1,2)\}]$. Since $b_{ij}$ depends only on $j$, we have $b_{ij} = \beta_{.j}$. Then it can be shown similarly as above that boundary solutions in this case are given by $\hat{\lambda}_{Y_{2}R_{2}}(j,2) = -\infty\Leftrightarrow\hat{\pi}_{+j+2} = 0\Leftrightarrow\hat{\beta}_{.j} = 0$ for at least one and at most $(J-1)$ values of $Y_{2}$. 
\vone\noindent
Finally, consider part 3. Under Model [M5], the log-linear parameters modelling the NMAR nechanisms of $Y_{1}$ and $Y_{2}$ are $\lambda_{R_{1}}(k)$, $\lambda_{R_{2}}(l)$, $\lambda_{Y_{1}R_{1}}(i,k)$ and $\lambda_{Y_{2}R_{2}}(j,l)$. The proof for the form of boundary solutions under Model [M5] follows on similar lines as for Models [M1]-[M4] shown above.

\subsection*{A2}{\it Proof of Theorem 5.1:} From Baker {\it et al.}~(1992), the MLE's $\hat{\alpha}_{i.}$ under the NMAR model for only $Y_{1}$ (Models [M1] and [M3]) satisfy
\begin{equation}\label{eq3.12}
\sum_{i}N\hat{\pi}_{ij11}\hat{\alpha}_{i.} = y_{+j21}, \quad \forall~1\leq j\leq I,
\end{equation}
while the MLE's $\hat{\beta}_{.j}$ under the NMAR model for only $Y_{2}$ (Models [M2] and [M4]) satisfy
\begin{equation}\label{eq3.13}
\sum_{j}N\hat{\pi}_{ij11}\hat{\beta}_{.j} = y_{i+12}, \quad \forall~1\leq i\leq I.
\end{equation}
The  MLE's $\hat{\alpha}_{i.}$ and $\hat{\beta}_{.j}$ under the NMAR model for both $Y_{1}$ and $Y_{2}$ (Model [M5]) satisfy both (\ref{eq3.12}) and (\ref{eq3.13}). Note that boundary solutions in Models [M1] and [M3] occur if $\hat{\alpha}_{i.} \leq 0$ for at least one and at most $(I - 1)$ values of $Y_{1}$, while boundary solutions in Models [M2] and [M4] occur if $\hat{\beta}_{.j} \leq 0$ for at least one and at most $(I - 1)$ values of $Y_{2}$. Also note that boundary solutions under [M5] occur if at least one of the following holds:
\begin{enumerate}
\item[(i)] $\hat{\alpha}_{i.} \leq 0$ for at least one and at most $(I-1)$ values of $Y_{1}$, 
\item[(ii)] $\hat{\beta}_{.j} \leq 0$ for at least one and at most $(I-1)$ values of $Y_{2}$.
\end{enumerate}
From (\ref{eq3.1}) and (\ref{eq3.12}), we have
\begin{equation*}
\nu(j,j') = \frac{y_{+j21}}{y_{+j'21}} = \frac{\sum_{i}\hat{\pi}_{ij11}\hat{\alpha}_{i.}}{\sum_{i}\hat{\pi}_{ij'11}\hat{\alpha}_{i.}},
\end{equation*} 
\begin{equation}\label{eq3.14}
\nu_{m}(j,j') - \nu(j,j') = \frac{\sum_{i\neq m_{1}}(\hat{\pi}_{m_{1}j11}\hat{\pi}_{ij'11} - \hat{\pi}_{m_{1}j'11}\hat{\pi}_{ij11})\hat{\alpha}_{i.}}{\hat{\pi}_{m_{1}j'11}\sum_{i}\hat{\pi}_{ij'11}\hat{\alpha}_{i.}},
\end{equation}
\begin{equation}\label{eq3.15}
\nu(j,j') - \nu_{n}(j,j') = \frac{\sum_{i\neq n_{1}}(\hat{\pi}_{n_{1}j'11}\hat{\pi}_{ij11} - \hat{\pi}_{n_{1}j11}\hat{\pi}_{ij'11})\hat{\alpha}_{i.}}{\hat{\pi}_{n_{1}j'11}\sum_{i}\hat{\pi}_{ij'11}\hat{\alpha}_{i.}},
\end{equation}
where $m_{1}$ and $n_{1}$ are the levels of $Y_{1}$ corresponding to $\nu_{m}(j,j')$ and $\nu_{n}(j,j')$ respectively. From (\ref{eq3.1}), we get
\begin{equation}\label{eq3.16}
\nu_{n}(j,j') = \frac{\hat{\pi}_{n_{1}j11}}{\hat{\pi}_{n_{1}j'11}} < \nu_{i}(j,j') = \frac{\hat{\pi}_{ij11}}{\hat{\pi}_{ij'11}} < \nu_{m}(j,j') = \frac{\hat{\pi}_{m_{1}j11}}{\hat{\pi}_{m_{1}j'11}}. 
\end{equation} 
From (\ref{eq3.16}), we have the following inequalities
\begin{equation}\label{eq3.17}
\hat{\pi}_{m_{1}j11}\hat{\pi}_{ij'11} > \hat{\pi}_{m_{1}j'11}\hat{\pi}_{ij11},~\hat{\pi}_{n_{1}j'11}\hat{\pi}_{ij11} > \hat{\pi}_{n_{1}j11}\hat{\pi}_{ij'11}~\textrm{for}~i\neq m_{1},n_{1}.
\end{equation}
Consider part (a). Suppose Condition 1 holds, which implies that (\ref{eq3.14}) and (\ref{eq3.15}) are of opposite signs. Using this fact and (\ref{eq3.17}), we observe that $\hat{\alpha}_{i.} < 0$ for at least one and at most $(I - 1)$ values of $Y_{1}$, that is, boundary solutions of the form $\hat{\pi}_{i+2+} = 0$ occur. 

Again from (\ref{eq3.2}) and (\ref{eq3.13}), we have
\begin{equation*}
\omega(i,i') = \frac{y_{i+12}}{y_{i'+12}} = \frac{\sum_{j}\hat{\pi}_{ij11}\hat{\beta}_{.j}}{\sum_{j}\hat{\pi}_{i'j11}\hat{\beta}_{.j}},
\end{equation*} 
\begin{equation}\label{eq3.18}
\omega_{m}(i,i') - \omega(i,i') = \frac{\sum_{j\neq m_{2}}(\hat{\pi}_{im_{2}11}\hat{\pi}_{i'j11} - \hat{\pi}_{i'm_{2}11}\hat{\pi}_{ij11})\hat{\beta}_{.j}}{\hat{\pi}_{i'm_{2}11}\sum_{i}\hat{\pi}_{i'j11}\hat{\beta}_{.j}},
\end{equation}
\begin{equation}\label{eq3.19}
\omega(i,i') - \omega_{n}(i,i') = \frac{\sum_{j\neq n_{2}}(\hat{\pi}_{i'n_{2}11}\hat{\pi}_{ij11} - \hat{\pi}_{in_{2}11}\hat{\pi}_{i'j11})\hat{\beta}_{.j}}{\hat{\pi}_{i'n_{2}11}\sum_{i}\hat{\pi}_{i'j11}\hat{\beta}_{.j}},
\end{equation}
where $m_{2}$ and $n_{2}$ are the levels of $Y_{2}$ corresponding to $\omega_{m}(i,i')$ and $\omega_{n}(i,i')$ respectively. From (\ref{eq3.2}), we get
\begin{equation}\label{eq3.20}
\omega_{n}(i,i') = \frac{\hat{\pi}_{in_{2}11}}{\hat{\pi}_{i'n_{2}11}} < \omega_{j}(i,i') = \frac{\hat{\pi}_{ij11}}{\hat{\pi}_{i'j11}} < \omega_{m}(i,i') = \frac{\hat{\pi}_{im_{2}11}}{\hat{\pi}_{i'm_{2}11}}. 
\end{equation} 
From (\ref{eq3.20}), we have the following inequalities
\begin{equation}\label{eq3.21}
\hat{\pi}_{m_{2}j11}\hat{\pi}_{ij'11} > \hat{\pi}_{m_{2}j'11}\hat{\pi}_{ij11},~\hat{\pi}_{n_{2}j'11}\hat{\pi}_{ij11} > \hat{\pi}_{n_{2}j11}\hat{\pi}_{ij'11}~\textrm{for}~j\neq m_{2},n_{2}.
\end{equation}
Now consider part (b). Assume Condition 2 holds, which implies that (\ref{eq3.18}) and (\ref{eq3.19}) are of opposite signs. Using this fact and (\ref{eq3.21}), we observe that $\hat{\beta}_{.j} < 0$ for at least one and at most $(I - 1)$ values of $Y_{2}$, that is, boundary solutions of the form $\hat{\pi}_{+j+2} = 0$ occur. 

Finally consider part (c). Assume at least one of Conditions 1 and 2 holds. The cases when only Condition 1 holds or only Condition 2 holds follow from the proofs of part (a) and part (b) respectively. So it is sufficient here to assume both Conditions 1 and 2 hold. This implies, from part (a), $\hat{\alpha}_{i.} < 0$ for at least one and at most $(I - 1)$ values of $Y_{1}$, that is, boundary solutions of the form $\hat{\pi}_{i+2+} = 0$ occur. Also from part (b), we have $\hat{\beta}_{.j} < 0$ for at least one and at most $(I - 1)$ values of $Y_{2}$, that is, boundary solutions of the form $\hat{\pi}_{+j+2} = 0$ occur.  
This completes the proof.

\subsection*{A3}{\it Proof of Theorem 5.2:} From Theorem \ref{th1}, the MLE's $\hat{\alpha}_{i.}$ and $\hat{\beta}_{.j}$ under Model [M5] satisfy
\begin{equation}\label{eq4.3}
\sum_{i}\hat{\mu}_{ij11}\hat{\alpha}_{i.} = y_{+j21} \quad \textrm{for}~j = 1,\ldots, I,
\end{equation}
\begin{equation}\label{eq4.4}
\sum_{j}\hat{\mu}_{ij11}\hat{\beta}_{.j} = y_{i+12} \quad \textrm{for}~i = 1,\ldots, I.
\end{equation}
Also, the MLE $\hat{\alpha}_{i.}$ under Models [M1] and [M3] satisfy (\ref{eq4.3}) only, while the MLE $\hat{\beta}_{.j}$ under Models [M2] and [M4] satisfy (\ref{eq4.4}) only. Note that boundary solutions under [M5] occur if at least one of the following conditions hold:
\begin{enumerate}
\item[(i)] $\hat{\alpha}_{i.} \leq 0$ for at least one and at most $(I-1)$ values of $Y_{1}$, 
\item[(ii)] $\hat{\beta}_{.j} \leq 0$ for at least one and at most $(I-1)$ values of $Y_{2}$.
\end{enumerate}
Also, boundary solutions in Models [M1] and [M3] are given by only Condition (i), while boundary solutions in Models [M2] and [M4] are given by only Condition (ii). In Lemma \ref{th3}, take $A = (\hat{\mu}_{ij11})$, ${\bf b} = (b_{j}) = (y_{+j21})$ and ${\bf b^{\ast}} = (b^{\ast}_{i}) = (y_{i+12})$ for $1\leq i\leq I,~1\leq j\leq I$. Then (\ref{eq4.3}) may be written as $A^{T}\alpha = {\bf b}$, while (\ref{eq4.4}) may be written as $A\beta = {\bf b^{\ast}}$, where $\alpha = (\alpha_{i.})$ and $\beta = (\beta_{.j})$. We prove Theorem \ref{th4} by contrapositive.

Consider part (a) first. Suppose Condition 1 in Theorem \ref{th4} does not hold. Then by Lemma \ref{th3}, $\alpha = (A^{T})^{-1}{\bf b} > {\bf 0}$. In other words, $\hat{\alpha}_{i.} > 0$ for all $1\leq i\leq I$, that is, boundary solutions under Models [M1] and [M3] do not occur. 

Consider part (b) now. Assume Condition 2 in Theorem \ref{th4} does not hold. Then by Lemma \ref{th3}, $\beta = A^{-1}{\bf b}^{\ast} > {\bf 0}$. In other words, $\hat{\beta}_{.j} > 0$ for all $1\leq j\leq I$, that is, boundary solutions under Models [M2] and [M4] do not occur. 

Finally consider part (c). Assume both Conditions 1 and 2 in Theorem \ref{th4} do not hold. Then by Lemma \ref{th3}, both $\hat{\alpha}_{i.} > 0$ and $\hat{\beta}_{.j} > 0$ for all $1\leq i\leq I,~1\leq j\leq I$, that is, boundary solutions under Model [M5] do not occur. 

Hence, the result follows.

\end{document}